\documentclass[10pt]{article}
\usepackage{dcolumn}
\usepackage{bm}
\usepackage{verbatim}       

\usepackage[pdftex]{graphicx}
\usepackage{amsthm}
\usepackage{amssymb}
\usepackage{undertilde}

\usepackage{amstext}
\usepackage{amsmath}
\usepackage{amsfonts}
\usepackage{units}
\usepackage{mathrsfs}
\usepackage{cite}
\newcommand{\p}{\partial}

\newcommand{\dd}{{\rm d}}

\newcommand{\bd}{\begin{definition}}                
\newcommand{\ed}{\end{definition}}                  
\newcommand{\bc}{\begin{corollary}}                 
\newcommand{\ec}{\end{corollary}}                   
\newcommand{\bl}{\begin{lemma}}                     
\newcommand{\el}{\end{lemma}}                       
\newcommand{\bp}{\begin{proposition}}            
\newcommand{\ep}{\end{proposition}}                
\newcommand{\bere}{\begin{remark}}                  
\newcommand{\ere}{\end{remark}}                     

\newcommand{\bt}{\begin{theorem}}
\newcommand{\et}{\end{theorem}}

\newcommand{\be}{\begin{equation}}
\newcommand{\ee}{\end{equation}}

\newcommand{\bit}{\begin{itemize}}
\newcommand{\eit}{\end{itemize}}
\newtheorem{theorem}{Theorem}[section]
\newtheorem{corollary}[theorem]{Corollary}
\newtheorem{lemma}[theorem]{Lemma}
\newtheorem{proposition}[theorem]{Proposition}
\theoremstyle{definition}
\newtheorem{definition}[theorem]{Definition}
\theoremstyle{remark}
\newtheorem{remark}[theorem]{Remark}
\newtheorem{example}[theorem]{Example}

\begin{document}




\title{On the regularity of Cauchy hypersurfaces and temporal functions in closed cone structures}

\author{E. Minguzzi\thanks{
Dipartimento di Matematica e Informatica ``U. Dini'', Universit\`a
degli Studi di Firenze, Via S. Marta 3,  I-50139 Firenze, Italy.
E-mail: ettore.minguzzi@unifi.it. } }

\date{}

\maketitle

\begin{abstract}
\noindent We complement our work on the causality of upper semi-continuous distributions of cones with some results on Cauchy hypersurfaces. We prove that every locally stably acausal Cauchy hypersurface is stable. Then we prove that the  signed distance $d_S$ from a spacelike hypersurface $S$ is, in a neighborhood of it, as regular as the  hypersurface, and by using this fact we give a proof that every Cauchy hypersurface is the   level set of a Cauchy temporal (and steep) function of the same regularity as the hypersurface.  We also show that in a globally hyperbolic closed cone structure compact spacelike hypersurfaces with boundary can be extended to Cauchy spacelike hypersurfaces of the same regularity. We end the work with a separation result and a density result.
\end{abstract}



%

\section{Introduction}

In a recent work \cite{minguzzi17} we developed the causality theory for ``closed cone structures'' i.e.\ upper semi-continuous distributions $p\mapsto C_p$ of closed sharp non-empty convex cones. This theory developed as a generalization of causality theory for Lorentzian geometry \cite{hawking73,minguzzi18b}, the previous works on conic distributions being mostly concerned with the existence of various type of increasing functions \cite{fathi12,bernard18}.

Since causality theory is  quite an extensive subject the work \cite{minguzzi17} could only cover some portions of it. Still it was unfortunate that some important topics were left aside.
Recently,  James Vickers and G\"unther H\"ormann asked me some questions on whether smooth spacelike Cauchy hypersurfaces are stable and whether they are level sets of smooth Cauchy temporal functions. These questions originated in their study of quantum field theory in globally hyperbolic  spacetime with non-$C^2$ metric \cite{hormann19}. A positive answer would generalize a well known theorem by Bernal and S\'anchez in Lorentzian geometry \cite[Thm.\ 1.2]{bernal06}.

The closely related question on whether globally hyperbolic closed cone structures admit  smooth  Cauchy temporal functions  was answered affirmatively in \cite{bernard18} and \cite{minguzzi17}  (the same question under continuity was answered in \cite{fathi12}), by using different smoothing techniques. This generalized another result by Bernal and S\'anchez \cite{bernal03}.

Bernard and Suhr used some smoothing techniques they developed in \cite{bernard18} that differ in philosophy from the approach introduced by Chru\'sciel, Grant and the author \cite{chrusciel13}, later refined in  \cite{minguzzi17}. Our approach consisted  in  first showing, through topological arguments, that the causal cones can be widened while preserving some key causality properties (e.g.\ global hyperbolicity or the Cauchy property of the hypersurface). This first  step makes the `room' needed for the subsequent step which consists in the smoothing of a suitable continuous increasing function. On the contrary Bernard and Suhr work out  the smoothing first, a step which is  more technically demanding, and then infer the stability properties more easily, but a posteriori, thanks to the obtained smoothness.

In a recent preprint Bernard and Suhr \cite{bernard19} show that the generalization of \cite[Thm.\ 1.2]{bernal06} can indeed be obtained via their smoothing technique, namely smooth spacelike Cauchy hypersurfaces are indeed level sets of smooth Cauchy temporal functions. It is natural to ask whether our approach \cite{minguzzi17} can also be used to answer this question. In this work we prove that it can and, as it turns out, the proof is conceptually simple as it consists in the single application of our smoothing theorem \cite[Thm.\ 3.9]{minguzzi17} to a signed distance function from the Cauchy hypersurface, cf.\ Thm.\ \ref{ngy}. Our analysis clarifies that in the non-smooth case the Cauchy temporal function  can be found of the same regularity as the hypersurface.

We start with the next section in which we prove  the equivalence between the locally stably acausal Cauchy hypersurfaces and the stable Cauchy hypersurfaces (Thm.\ \ref{sta}). This result, being  propaedeutic for the proof of the other results, is also quite interesting because it applies  to general Lipschitz hypersurfaces, not just to smooth ones. Then we introduce the notion of signed distance and obtain a result on its regularity (Thm.\ \ref{nja}) that is analogous to a not so well known result in Riemannian geometry. We show how to obtain the Cauchy temporal function  adapted to the Cauchy hypersurface (Thm.\ \ref{ngy}). Next, we prove a generalization of \cite[Thm.\ 1.1]{bernal06} namely that in a globally hyperbolic closed cone structure compact spacelike hypersurfaces with boundary can be extended to Cauchy spacelike hypersurfaces of the same regularity (Thm.\ \ref{ngz}). Finally, we prove a smooth separation result (Thm.\ \ref{mos}) and a density result (Thm.\ \ref{jif}) for stably causal closed cone structures.

Our conventions, terminology and notations are the same as in  \cite{minguzzi17}. The reader is assumed to be familiar with that work, particularly with the definition of closed (and proper) cone structure, and with the definition of closed (and proper) Lorentz-Finsler space.



\section{Stability, signed distance and smoothing}

We start with the topological proof that cones can be widened preserving the Cauchy property. In this way we shall make `the room' needed for the subsequent smoothing by convolution.

\subsection{Stability and Cauchy hypersurfaces}

We need a notion that replaces that of spacelikeness in a rough setting.
\begin{definition}
Let $(M,C)$ be a closed cone structure. A topological hypersurface
$S$ is {\em locally stably acausal} if every point $p\in S$ admits an open neighborhood $U$ and a locally
Lipschitz proper cone structure $C'>C$ on $U$, such that $S\cap U$ is acausal in $(U,C')$.
\end{definition}

On a  closed cone structure $(M,C)$ a $C^1$ hypersurface $S$ is spacelike  if for every $p\in S$, $T_pS\cap C_p=\emptyset$ (remember that $C_p$ does not include the zero vector). The $C^1$ spacelike hypersurfaces are the simplest examples of locally stably acausal hypersurfaces.

We recall some definitions from \cite{minguzzi17}.

\begin{definition}
Let $(M,C)$ be a closed cone structure. A Cauchy hypersurface is
an acausal topological hypersurface $S$ such that $D(S) = M$ (namely every inextendible continuous causal curve intersects $M$). A stable Cauchy hypersurface is a Cauchy hypersurface for $(M,C')$ where $C' > C$ is a locally Lipschitz
proper cone structure.
\end{definition}

Notice that if $C$ is a closed cone structure and $C'$ is a locally Lipschitz proper cone structure, then we can find an intermediate locally Lipschitz proper cone structure $C''$, $C<C''<C'$, cf.\ \cite[Thm.\ 2.23]{minguzzi17}. As a consequence, if $S$ is a stable Cauchy hypersurface for $(M,C)$ then we can find a locally Lipschitz proper cone structure $C''>C$, such that $S$ is a {\em stable} Cauchy hypersurface for $(M,C'')$.

Our first objective is to show that the locally stably acausal Cauchy hypersurfaces coincide with the stable Cauchy hypersurfaces.

By definition a closed cone structure  is globally hyperbolic if it is acausal and the causally convex hull of compact sets is compact \cite{minguzzi19c} (this definition improves the previous ones discussed in \cite{bernard18,minguzzi17}).

A closed cone structure that admits a Cauchy hypersurface is globally hyperbolic \cite[Thm.\ 2.43]{minguzzi17}. The converse holds as well \cite[Thm.\ 2.42]{minguzzi17} \cite{bernard18}, but we shall not use this fact, for we shall provide a new simple proof later on, cf.\ Thm.\ \ref{jao}.

\begin{lemma} \label{mfp}
Let $(M,C)$ be a closed cone structure. Let  $S$ be a locally stably acausal Cauchy hypersurface, then $J^+(S)$ and $J^-(S)$ are closed, $J^+(S)\cap J^-(S)=S$, $J^+(S)\cup J^-(S)=M$. The set $M\backslash S$ is disconnected, for it is the union of the disjoint non-empty open sets $J^+(S)\backslash S=M\backslash J^-(S)$ and $J^-(S)\backslash S=M\backslash J^+(S)$.
\end{lemma}

\begin{proof}
Through every point there passes an inextendible continuous causal curve intersecting $S$, thus $J^+(S)\cup J^-(S)=M$. The acausality of $S$ implies $J^+(S)\cap J^-(S)=S$.

Suppose that there is a point $p\in \overline{J^+(S)}\backslash J^+(S)$ then necessarily $p\in J^-(S)\backslash S$.
 Since $p\in \overline{J^+(S)}\backslash J^+(S)$ by a standard limit curve argument \cite[Thm.\ 2.14]{minguzzi17} there is a past inextendible continuous causal curve $\eta$ ending at $p$ and not intersecting $S$.
 There must exist a compact neighborhood $K\ni p$ such that every future inextendible continuous causal curve issued from $K$ reaches $S$. Indeed, suppose not, then we can find a sequence of future inextendible continuous causal curves $\sigma_n$ with staring point $p_n\to p$, not intersecting $S$. By the limit curve theorem there is a limit future inextendible continuous causal curve $\sigma$ to which a subsequence (here denoted in the same way) $\sigma_n$ converges (uniformly on compact subsets for a suitable parametrization, see details in \cite[Thm.\ 2.14]{minguzzi17}). If $\sigma$ intersects (and hence crosses) $S$ then so do $\sigma_n$ for sufficiently large $n$ (due to \cite[Thm.\ 2.33]{minguzzi17}), which is a contradiction. If $\sigma$ does not intersect $S$, then $\sigma \circ \eta$ provides an inextendible continuous causal curve passing through $p$ not intersecting $S$, a contradiction.
 Thus $K$ exists, but then since $p\in \overline{J^+(S)}$ we can find $q\in K\cap J^+(S)$. Once any future inextendible continuous causal curve starting from $q$ is chosen  (which exists by \cite[Thm.\ 2.1]{minguzzi17}), it intersects $S$, which proves that $S$ is not acausal. The contradiction proves that $J^+(S)$ is closed.

Since through every point of $S$ there passes an inextendible continuous causal curve the sets  $J^\pm(S)\backslash S$ are indeed non-empty thus $M\backslash S$ is indeed disconnected as stated.
\end{proof}

\begin{lemma} \label{mfo}
Let $(M,C)$ be a closed cone structure. Let  $S$ be a locally stably acausal Cauchy hypersurface and let $\tilde C>C$ be  a locally
Lipschitz proper cone structure  chosen  with so narrow cones  that every point $p\in S$ admits a neighborhood $U$ such that $S\cap U$ is acausal in $(U,\tilde C)$. Then $S$ is acausal in $(M,\tilde C)$.
\end{lemma}

\begin{proof}
Suppose that there is a $\tilde C$-causal curve $\gamma\colon [0,1]\to M$, intersecting $S$ just at the endpoints.
The sets $M\backslash J^\pm(S)$ are open, disjoint, and $[M\backslash J^+(S)]\cup [M\backslash J^-(S)]=M\backslash S$, thus $\gamma^{-1}(M\backslash J^+(S))$ and $\gamma^{-1}(M\backslash J^-(S))$ are two open disjoint sets whose union gives $(0,1)$.
By the local stable acausality of $S$ there are $0<\alpha<\beta<1$ such that $\gamma((0,\alpha))\subset J^+(S)\backslash S \subset M\backslash J^-(S)$ and $\gamma((\beta,1))\subset J^-(S)\backslash S\subset M\backslash J^+(S)$ (see \cite[Thm.\ 2.33]{minguzzi17}) which proves that $\gamma^{-1}(M\backslash J^+(S))$ and $\gamma^{-1}(M\backslash J^-(S))$ are non-empty
 in contradiction with the connectedness of $(0,1)$.
 The contradiction proves that $\gamma$ does not exist.
\end{proof}

The next result was almost completely proved in \cite[Thm.\ 2.41]{minguzzi17}, where it was shown that the cones can be widened preserving the property that every inextendible continuous causal curve intersects $S$. Unfortunately, the proof that $S$ was acausal was missing.

\begin{theorem} \label{sta}
Let $(M,C)$ be a closed cone structure. The locally stably acausal Cauchy hypersurfaces coincide with the stable Cauchy hypersurfaces.
\end{theorem}

\begin{proof}
To the left the implication is clear. For the other direction, by Lemma \ref{mfo} the  locally stably acausal Cauchy hypersurface $S$ is actually stably acausal, thus by  \cite[Thm.\ 2.41]{minguzzi17} it is a stable Cauchy hypersurfaces.
\end{proof}

\begin{example}
Not all Cauchy hypersurfaces are locally stably acausal: consider Minkowski 1+1 spacetime of metric $g=-\dd t^2+\dd x^2$ and the Cauchy hypersurface $S$ given by the graph of the function $t=\tanh x$.  The hypersurface $S$ is not locally stably acausal in a neighborhood of $(0,0)$. Notice that $S$ is $C^1$ but the tangent space at the origin is null.
\end{example}

Proposition 2.22 in \cite{minguzzi17} now simplifies as follows

\begin{proposition}
Let $(M,C)$ be a closed cone structure.
Any two locally stably acausal $C^k$, $0\le k \le \infty$, Cauchy hypersurfaces are $C^k$ diffeomorphic.
\end{proposition}





\subsection{The signed distance function and its regularity} \label{nog}

The proof of the next result is similar to the Lorentzian one.
\begin{proposition}
Let $(M,C)$ be a closed cone structure and let $S$ be a locally stably acausal Cauchy hypersurface. Then for every compact subset $K$, $J^-(K)\cap J^+(S)$
is compact.
\end{proposition}

\begin{proof}
Otherwise there is a sequence of continuous casual curves $\sigma_n$ connecting $S$ to $K$ and points $p_n\in \sigma_n$, $p_n\to \infty$ (i.e. escaping every compact set). A standard  limit curve argument \cite[Thm.\ 2.14]{minguzzi17} would produce a past inextendible continuous causal curve $\sigma$ ending at $K$. Thus by the stable local acausality of $S$ it would cross $S$
entering, \cite[Thm.\ 2.33]{minguzzi17}, the open set $J^-(S)\backslash S$. But then a subsequence
$\sigma_k \to \sigma$ starting from $S$, would have to enter it, contradicting the acausality
of $S$.
\end{proof}

\begin{proposition}
Let $(M,\mathscr{F})$ be a locally Lipschitz proper Lorentz-Finsler space such that $\mathscr{F}(\p C)=0$, and let $S$ be a locally stably acausal Cauchy hypersurface. The functions $d^+_S\colon J^+(S) \to [0,+\infty)$, $d^+_S(p):=d(S,p)$, and the function $d^-_S\colon J^-(S) \to (-\infty,0]$, $d^-_S(p):=-d(p,S)$ are finite and continuous. Similarly the function $d_S$ which coincides with $d^+_S$ on $J^+(S)$ and with $d^-_S$ on $J^-(S)$, is finite and continuous.
\end{proposition}

The function $d_S$ is called the {\em  signed (Lorentz-Finsler) distance function}.

\begin{proof}
By the upper semi-continuity of the length functional \cite[Thm.\ 2.54]{minguzzi17} through a standard limit curve argument one gets that the distance is realized, namely there is a continuous causal curve connecting $S$ to $p$ whose Lorentz-Finsler length coincides the distance $d(S,p)$. In particular $d^+_S$ is finite. The same argument shows that $d^+_S$ is upper semi-continuous. For the lower semi-continuity one has to use the locally Lipschitz property of the cone distribution. The proof is analogous to that of \cite[Thm.\ 2.53]{minguzzi17}.
\end{proof}

The next result proves that the signed distance is a {\em rushing time} in the sense of \cite{minguzzi18b}.
This anti-Lipschitz property will be of fundamental importance for the application of our smoothing theorem.

\begin{proposition}
Let $(M,\mathscr{F})$ be a locally Lipschitz proper Lorentz-Finsler space such that $\mathscr{F}(\p C)=0$, and let $S$ be a locally stably acausal Cauchy hypersurface. The continuous and finite function $d_S\colon M\to \mathbb{R}$ satisfies: for every $(p,q)\in J$,
\[
d_S(q)-d_S(p)\ge d(p,q).
\]
\end{proposition}

\begin{proof}
If $p$ and $q$ belong to $J^+(S)$, then $d_S=d_S^+=d(S,\cdot)$, and we have for every $r\in J^{-}(p)\cap S$, $d(S,q)\ge d(r,q)\ge  d(r,p)+d(p,q)$, which taking the supremum over $r$ gives,
$d(S,q)\ge d(S,p)+d(p,q)$ which is what we wanted to prove.
If $p$ and $q$ belong to $J^-(S)$, then $d_S=d_S^-=-d(\cdot, S)$, and we have for every $r\in J^+(q)\cap S$,  $d(p,S)\ge d(p,r)\ge d(p,q)+d(q,r)$, which taking the supremum over $r$ gives $d(p,S)\ge  d(p,q)+d(q,S)$, which written in the form $[-d(q,S)]-[-d(p,S)]\ge d(p,q)$ gives what we wanted to prove. For $p\in J^-(S)\backslash S$ and $q\in J^+(S)\backslash S$ consider the maximizing continuous causal curve $\sigma$ connecting $p$ to $q$ and let $r$ be its intersection point with $S$. Then $d(p,q)=\ell(\sigma)=d(p,r)+(r,q)$. Now $d(S,q)\ge d(r, q)$ and $d(p,S)\ge d(p,r)$ thus $d(S,q)+d(p,S)\ge d(p,r)+d(r,q)=d(p,q)$, which can be rewritten as $d_S(q)-d_S(p)\ge d(p,q)$.
\end{proof}

Let $C$ be a smooth distribution of proper sharp strongly convex closed cones.
Let us consider a positive homogeneous function ${\mathscr{F}}\colon C\to [0,+\infty)$
 such that ${\mathscr{F}}^{-1}(0)=\p C$, ${\mathscr{L}}:=-{\mathscr{F}}^2/2$ is such that the vertical Hessian on $ C$, $\dd^2_y \mathscr{L}$, is smooth and Lorentzian, and $\dd {\mathscr{L}}\ne \emptyset$ on $\p C$. This function exists by \cite[Prop.\ 13]{minguzzi15e}, see also  \cite[Cor.\ 5.8]{javaloyes18}. With these properties
  $(M,{\mathscr{F}})$ is certainly a  locally Lipschitz proper Lorentz-Finsler space \cite[Thm.\ 2.52]{minguzzi17} so the previous results hold true.

We want to know how the regularity $C^k$ of a spacelike Cauchy hypersurface $S$ reflects itself on the regularity of the signed distance $d_S$. For this reason we are considering a quite regular Lorentz-Finsler structure, in this way all the non-differentiability aspects of $d_S$ will be due to the hypersurface. Notice that under this regularity the timelike geodesics make sense and they coincide with the locally maximizing continuous causal curve with non-vanishing length \cite[Thm.\ 6]{minguzzi13d} (the formalism of \cite{minguzzi13d} can be applied because the Finsler Lagrangian can be extended to all of $TM\backslash 0$, see \cite{minguzzi14h}). Thus the maximal curves that realize the distance $d^+_S$ or $-d^-_S$ are timelike geodesics normal to $S$ (all our continuous causal curves are future-directed).

Here normality has to be understood as follows. A vector $v$ is normal to $S$ if it can be written in the form $v=r n$, where $r\in \mathbb{R}$ and $n$ is a (future-directed) normalized timelike vector $n\in \textrm{Int} C$ (necessarily unique) such that $\mathscr{F}(n)=1$, $TS=\textrm{ker} g_n(n,\cdot)$.

The normal bundle $N\subset TM$ can be identified with the image of a map $S\times \mathbb{R}\to TM$, $(p,r)\mapsto  n_p r$, where $p\mapsto n_p$ is the normalized section. Since this section is $C^{k-1}$, $N$ is a $C^{k-1}$ submanifold of $TM$. Stated in another way $S\times \mathbb{R}$ with $N$ are $C^{k-1}$ diffeomorphic.
For $v\in N$, we define $\exp_S v=\gamma_n(r)$, where $\gamma_v$ is the timelike geodesic with initial condition $v$.
For $v\in \textrm{Int} C$ let $\gamma_v$ be the timelike geodesic with initial condition $v$.
Remember that the geodesic flow  $\textrm{Int} C\times \mathbb{R} \to \textrm{Int} C$, $(v,t)\mapsto \gamma_v'(t)$, is smooth because the spray derived from the Finsler Lagrangian is smooth. Thus the map $S\times \mathbb{R} \to TM$, $(p,t)\mapsto \gamma_{n_p}'(t)$ is $C^{k-1}$ as it is the composition of a $C^{k-1}$ map and of a smooth map.

Notice that $\gamma_n(r)=\pi\circ\gamma_n'(r)$, where $\pi\colon TM\to M$. It is well known that the map $(n,r)\mapsto \gamma_n(r)$ is $C^{k-1}$ and locally invertible (with the identification of $S\times\mathbb{R}$ with $N$ this map is basically $\exp_S$), in the sense that there is an open neighborhood $V$ of the zero section of $S\times \mathbb{R}$ and an open neighborhood $U$ of $S$ that are $C^{k-1}$ diffeomorphic according to the map.
The regularity of the map that assigns to $q\in U$, $q=\pi(\gamma_n'(r))$ the vector $w_q=\gamma_n'(r)\in T_qM$ is $C^{k-1}$ because one has first to apply the $C^{k-1}$ inverse of the map $(n,r)\mapsto \gamma_n(r)$ to find $(n,r)$ and then to apply the $C^{k-1}$ map $(n,r)\mapsto \gamma_n'(r)$ to find $w_q$.


The next theorem generalizes a well known result on the regularity of the signed distance function in Riemannian spaces \cite{krantz81,matsumoto92}.

\begin{theorem} \label{nja}
Let $(M,\mathscr{F})$ be a sufficiently regular proper Lorentz-Finsler space as described above. If $S$ is a $C^{1,1}$ or a $C^k$, $2\le k\le \infty$, spacelike Cauchy hypersuface, then $d_S$ has the same regularity  in a neighborhood $U$ of $S$.
\end{theorem}

Weaker regularity options for the hypersurface could be considered, but one would need to introduce the notion of spacelike hypersurface of positive reach.

Due to Thm.\ \cite[Thm.\ 2.43]{minguzzi17} the Cauchy property can be dropped.

\begin{proof}
This fact follows from the previous discussion and from  the next formula  (analogous to Eq.\ (13) in \cite{minguzzi13d}, see also Remark 7 in the same reference). At $q\in U$
\[
\dd (d_S)(q)= g_{w_q}(w_q, \cdot)=\dd \mathscr{L}(w_q)(\cdot).
\]
Since $q\mapsto w_q$ is $C^{k-1}$ and non-zero, and since $g$ (or $\dd \mathscr{L}$) is smooth over the slit tangent bundle we have that $d_S$ is $C^k$ (in the $C^{1,1}$ case the argument is similar but $\exp_S$ really establishes a lipeomorphism not a diffeomorphism  \cite[Thm.\ 13]{minguzzi13d} \cite{kunzinger13}).
\end{proof}

\subsection{Smoothing signed distance functions}

We recall our smoothing theorem \cite[Thm.\ 3.9]{minguzzi17} \cite{chrusciel13} which we slightly improve by adding the penultimate statement.

\begin{theorem} \label{moz}
  Let $({ M},C)$ be a closed cone structure and
 let $\tau\colon M\to \mathbb{R}$ be a continuous function. Suppose that there is a $C^0$ proper cone structure $\hat C>C$ and continuous functions  homogeneous of degree one in the fiber $\underline F, \overline F\colon \hat C\to \mathbb{R}$ such that for every  $\hat C$-timelike  curve $x\colon [0,1]\to M$
 \begin{equation}
 \int_x \underline F(\dot x) \dd t\le \tau(x(1))-\tau(x(0))\le \int_x \overline F(\dot x) \dd t.
 \end{equation}
 Let $h$ be an arbitrary Riemannian metric, then for every continuous function $\alpha\colon { M} \to (0,+\infty)$ there exists
a smooth  function $\hat{\tau}$ such that $\vert \hat\tau-\tau\vert <\alpha$ and for every $v\in C$
\begin{equation} \label{kid}
\underline F(v)- \Vert v\Vert_h \le \dd \hat \tau(v) \le \overline F(v)+ \Vert v\Vert_h .
\end{equation}
Let $A$ be a closed set and $G\supset A$ an open neighborhood. If $\tau$ is already $C^1$ and satisfying (\ref{kid}) on $G$ then $\hat \tau$ can be found  such that it is $C^1$ and (\ref{kid}) holds everywhere, $\vert \hat\tau-\tau\vert <\alpha$ on $M$,  $\hat \tau$  coincides with $\tau$ on $A$, $\hat \tau$ has the regularity of $\tau$ on $G$, and is smooth outside $G$.

Similar versions, in which some of the functions $\underline F, \overline F$ do not exist hold true. One has just to drop the corresponding inequalities in (\ref{kid}).
\end{theorem}

\begin{proof}
We added just the penultimate sentence with respect to the previous version. In order to prove it we modify the original proof by taking, for $p\notin A$, $\epsilon(p)$ so small that  $B_p(3\epsilon(p))\cap A=\emptyset$, while for $p\in G$, $\epsilon(p)$ so small that $B_p(3\epsilon(p))\subset G$. Recall that in the original proof by using $\sigma$-compactness we pass to a locally finite covering of $M$ given by the open sets $\mathscr{O}_i=B_{p_i}(\epsilon_i)$, $\epsilon_i=\epsilon(p_i)$, we introduce a partition of unity $\varphi_i$ subordinate to the covering $\{\mathscr{O}_i\}$,   in $\mathscr{O}_i$ we  smooth  $\tau$ getting a function  $\tau_i$, and define $\hat \tau=\sum_i\varphi_i\tau_i$. Instead of always constructing $\tau_i$ through  convolution with a smooth kernel $\chi$, it is sufficient to set $\tau_i=\tau$ whenever $p_i\in G$, while define $\tau_i$ as in the original proof for $p_i\notin G$. Notice that each $\tau_i$ is at least as regular as $\tau$, so the same can be said of $\hat \tau$. Given the support of $\varphi_i$, $\hat \tau$ is  smooth outside $G$ and coincides with $\tau$ on $A$. The calculation on p.\ 110 of the original proof easily passes through with this modification, so Eq.\ (\ref{kid}) does indeed hold everywhere.
\end{proof}

The following results hold also for closed cone structures $(M,C)$ just set $\mathscr{F}=0$.
They are obtained by smoothing the signed distance function $d'_S$ for a suitable $\mathscr{F}'>\mathscr{F}$.

The first result is well known \cite[Thm.\ 2.45,3.12]{minguzzi17} \cite{fathi12,bernard18}. Compared to the proof in \cite{minguzzi17} here we use the signed distance to produce the anti-Lipschitz function to be smoothed, while there we made use of a product trick (which gave also information on stable spacetimes). We include this proof because it shows the effectiveness of the smoothing theorem coupled with the signed distance idea.

\begin{theorem} \label{jao}
Let $(M,\mathscr{F})$ be  a closed Lorentz-Finsler space. The following properties are equivalent
\begin{itemize}
\item[(a)] Global hyperbolicity.
\item[(b)] There is a Cauchy hypersurface.
\item[(c)] There is a Cauchy time function.
\end{itemize}
Moreover, in this case the Cauchy hypersuface can be found stable, and for any chosen complete Riemannian metric $h$ the Cauchy time function can be found smooth temporal $h$-steep and strictly $\mathscr{F}$-steep.
\end{theorem}

\begin{proof}
$(b)\Rightarrow (a)$. This direction is proved in \cite[2.43]{minguzzi17} with a proof similar to the Lorentzian one.

$(c)\Rightarrow (b)$. Clear, just take a level set of the Cauchy time function.

$(a)\Rightarrow (b)$ and $(c)$. This is proved in  \cite[2.42]{minguzzi17}: global hyperbolicity is stable, as can be proved through a direct topological proof \cite[Thm.\ 2.39]{minguzzi17}, thus there is a locally Lipschitz proper cone structure $C'>C$ such that $(M,C')$ is globally hyperbolic. The Geroch volume function for $(M,C')$ provides a Cauchy time function for $(M,C)$ and its level sets are stable Cauchy hypersurfaces for $(M,C)$.

Let $\mathscr{F}'$ be chosen such that $\mathscr{F}'(\p C')=0$, and sufficiently large so that $\mathscr{F}'> 2 \mathscr{F}$ on $C$, and such that $\mathscr{F}'{}^{-1}(1)\cap C$,  is contained in the unit ball of $4h$. As a consequence, for every $v\in C$, $2 \Vert v \Vert_h \le  \mathscr{F}'(v)$.

Given a  Cauchy hypersurface $S$ for $(M,C')$, the function $t\colon M\to \mathbb{R}$, $t=d'_S$, satisfies over every $C'$-timelike curve $x\colon [0,1]\to M$
\begin{equation} \label{jpo}
\int \mathscr{F}'(\dot x) \dd t=\ell'(x) \le  t(x(1))-t(x(0)).
\end{equation}
By Theorem \ref{moz} there is a smooth function $\tau$, $\vert \tau-t\vert<1$, that satisfies  $\dd \tau(v)\ge \mathscr{F}'(v)-\Vert v \Vert_h$ for every $v\in C$. But $\mathscr{F}'(v)- \Vert v\Vert_h \ge \mathscr{F}'(v)/2 > \mathscr{F}(v)$ and $\mathscr{F}'(v)- \Vert v\Vert_h \ge \Vert v \Vert_h$. The Cauchy property follows from $h$-steepness and the fact that $h$ is complete.
\end{proof}

\begin{theorem} \label{ngy}
Let $(M,\mathscr{F})$ be a closed Lorentz-Finsler space and let $h$ be a complete Riemannian metric. Let $S$ be a $C^{1,1}$ or a $C^k$, $2\le k\le \infty$,  spacelike   Cauchy hypersurface. There is an open neighborhood $U\supset S$, such that for any chosen open neighborhood $V\subset U$, $V\supset S$, there is  an $h$-steep and strictly $\mathscr{F}$-steep (hence temporal)  Cauchy function $\tau$, $\tau^{-1}(0)=S$, which has the same regularity as $S$ in $V$ and is smooth elsewhere.
\end{theorem}

\begin{proof}
The hypersurface $S$ is locally stably acausal and hence a stable Cauchy hypersurface.
Let $C'$ be a smooth strongly convex proper cone structure  such that $C'>C$ is globally hyperbolic and $S$ is a  stable Cauchy hypersurface for  $C'$. As mentioned in the previous section we can find a  smooth positive homogeneous function $\tilde{\mathscr{F}}\colon C'\to [0,+\infty)$
 such that $\tilde{\mathscr{F}}(\p C')=0$, $\tilde{\mathscr{L}}:=-\tilde{\mathscr{F}}^2/2$  has Lorentzian  vertical Hessian on $\textrm{Int} C'$, and $\dd \tilde{\mathscr{L}}\ne \emptyset$ on $\p C'$, so that $(M,\tilde{\mathscr{F}})$ is a locally Lipschitz proper Lorentz-Finsler space. Let $\mathscr{F}'=\varphi\tilde{\mathscr{F}}$ and choose the factor $\varphi\colon M\to (0,+\infty)$ so that $\mathscr{F}'$ is so large that $\mathscr{F}'> 2 \mathscr{F}$ on $C$, and  such that $\mathscr{F}'{}^{-1}(1)\cap C$ is contained in the unit ball of $4h$.
As a consequence, for every $v\in C$, $2 \Vert v \Vert_h \le  \mathscr{F}'(v)$. Notice that $\mathscr{F}'$ retains all the mentioned properties of $\tilde{\mathscr{F}}$. The function $t\colon M\to \mathbb{R}$, $t=d'_S$, satisfies Eq.\ (\ref{jpo}) over every $C'$-timelike curve $x\colon [0,1]\to M$.
By Thm.\ \ref{nja} there is a neighborhood $U$ of $S$ such that $t\vert_U$ has the same regularity as $S$.
In particular, $t\vert_U$ is $C^1$ thus using the inequality (\ref{jpo})  we have on $U$, (just consider an $\mathscr{F}'$-geodesic with tangent $v$)  $\dd t(v)\ge \mathscr{F}'(v)$ for every  $C'$-timelike vector $v$, and hence the inequality holds  for every $v\in C$.  We conclude that $\dd t(v)\ge \mathscr{F}'(v)-\Vert v \Vert_h$  for every  $v\in C\vert_U$.
Let $A$, $V$, be  closed and open neighborhoods of $S$, respectively, such that  $A\subset V\subset U$. Let $\alpha\colon M\to \mathbb{R}$ be a positive continuous function such that $\alpha \vert_{M\backslash A} <\vert t\vert_{M\backslash A}\vert $. By Theorem \ref{moz} there is a continuous function $\tau$, $\vert \tau-t\vert<\alpha$, that coincides with $t$ on $A$, has the regularity of $t$ and hence of $S$ on $V$, is smooth outside $V$,
 and satisfies  $\dd \tau(v)\ge \mathscr{F}'(v)-\Vert v \Vert_h$ for every $v\in C$. But $\mathscr{F}'(v)- \Vert v\Vert_h \ge \mathscr{F}'(v)/2 > \mathscr{F}(v)$ and $\mathscr{F}'(v)- \Vert v\Vert_h \ge \Vert v \Vert_h$. The Cauchy property follows from $h$-steepness and the fact that $h$ is complete. Finally, since $\tau$ coincides with $t$ on $A$, and has the same sign of $t$ outside $A$, we have $\tau^{-1}(0)= S$.
\end{proof}

The following result is analogous to the previous one for rough $S$.

\begin{theorem} \label{koh}
Let $(M,\mathscr{F})$ be a closed Lorentz-Finsler space and let $h$ be a complete Riemannian metric.
 Every stable Cauchy hypersurface $S$ is the zero level set of a Cauchy time function $\tau$, which is smooth $h$-steep and strictly $\mathscr{F}$-steep (hence temporal)  outside $S$.
\end{theorem}


\begin{proof}
Let $(M,\mathscr{F}')$ be any locally Lipschitz proper Lorentz-Finsler space such that $C'>C$ is globally hyperbolic and $S$ is a  stable Cauchy hypersurface for  $C'$. Moreover let $\mathscr{F}'$ be chosen such that $\mathscr{F}'(\p C')=0$, and sufficiently large so that $\mathscr{F}'> 2 \mathscr{F}$ on $C$, and such that $\mathscr{F}'{}^{-1}(1)\cap C$,  is contained in the unit ball of $4h$. As a consequence, for every $v\in C$, $2 \Vert v \Vert_h \le  \mathscr{F}'(v)$. The function $t\colon M\to \mathbb{R}$, $t=d'_S$, satisfies Eq.\ (\ref{jpo}) over every $C'$-timelike curve $x\colon [0,1]\to M$.
Consider the spacetime $J^+(S)\backslash S$ and let $\alpha\colon J^+(S)\backslash S \to (0,+\infty)$,  be a positive continuous function such that for every $q\in S$,  $\alpha(p)\to 0$ for $p\to q$. By  \cite[Thm.\ 3.9]{minguzzi17} there is a function $\tau^+\colon J^+(S)\backslash S \to \mathbb{R}$ such that $\vert \tau^+-t \vert<\alpha$ and for every $v\in C$, $\mathscr{F}'(v)-\Vert v \Vert_h\le \dd \tau(v)$. But $\mathscr{F}'(v)- \Vert v\Vert_h \ge \mathscr{F}'(v)/2 > \mathscr{F}(v)$ and $\mathscr{F}'(v)- \Vert v\Vert_h \ge \Vert v \Vert_h$. Similarly, we obtain $\tau^-$ and then the function $\tau$ equal to zero on $S$, to $\tau^+$ on
$J^+(S)\backslash S$ and to $\tau^-$ on $J^-(S)\backslash S$ has the desired properties (the Cauchy property follows from $h$-steepness and the fact that $h$ is complete).
\end{proof}

A stable Cauchy time function for a closed cone structure $(M,C)$ is a function which is a Cauchy time function for some locally Lipschitz proper cone structure $C'>C$, and hence for $C$.

\begin{corollary}
The stable Cauchy hypersurfaces are the zero level sets of the stable Cauchy time functions.
\end{corollary}

\begin{proof}
If $S$ is a stable Cauchy hypersurface for $(M,C)$ then it is a Cauchy hypersurface for some locally Lipschitz proper cone structure $\tilde C>C$, and hence a stable Cauchy hypersurface for some locally Lipschitz proper cone structure  $C'$, intermediate between the previous ones $C<C'<\tilde C$, cf.\ \cite[Thm.\ 2.23]{minguzzi17}. By Thm.\ \ref{koh} there is a Cauchy time function $t$ for $(M,C')$ such that $S=t^{-1}(0)$.

Let $t$ be a stable Cauchy time function, then there is a locally Lipschitz proper cone structure $C'>C$ such that $t$ is a Cauchy time function for $C'$. In particular, $t^{-1}(0)$ is a Cauchy hypersurface for $C'$. Thus the zero level set of a stable Cauchy time function is a stable Cauchy hypersurface.
\end{proof}

In the following  result by $C^k$ {\em closed acausal spacelike hypersurface with boundary} we understand an acausal spacelike hypersurface $S$ which is a closed set for the topology of $M$ (and hence includes its edge) and can be slightly extended to an acausal spacelike hypersurface $\tilde S\supset S$, which is a $C^k$ submanifold,  does not comprise its edge, and is such that $\textrm{edge} \tilde S\cap \textrm{edge} S=\emptyset$. In other words $S$ can be slightly enlarged so as to comprise the previous edge in the new relative interior points but maintaining the causality and differentiability properties.

\begin{theorem} \label{ngz}
Let $(M,C)$ be a globally hyperbolic closed cone structure.
Let $S$  be a $C^{1,1}$ or a $C^k$, $2\le k\le \infty$,   compact acausal spacelike hypersurface with boundary. There is an open neighborhood $U\supset S$ such that for any chosen open neighborhood $V \subset U$, $V\supset S$, there is a  spacelike Cauchy hypersurface $\Sigma \supset S$, which has the regularity of $S$ in $V$ and is smooth elsewhere.
\end{theorem}


This result holds also for Lipschitz regularity of $S$, provided $S$ is  locally stably acausal (rather than spacelike) and provided it can be slightly extended over the edge preserving its causality properties. The proof is obtained through a minimal modification of the following one.

Some compactness assumption is necessary. Partial Cauchy hypersurfaces are not extendable to Cauchy hypersurfaces in general, consider the hyperboloid $t=-(1+x^2+y^2+z^2)^{1/2}$ in Minkowski spacetime.

\begin{proof}
 Let $\hat S$, $S\subset \hat S\subset S$ be a compact acausal spacelike hypersurface with boundary, such that $\textrm{edge} \ S\cap \textrm{edge} \hat S=\emptyset$.
By the stability of global hyperbolicity we can find a locally Lipschitz proper cone structure $\tilde C>C$ such that $(M,\tilde C)$ is globally hyperbolic.
Moreover, we can find a locally Lipschitz proper cone structure $C'$, $C<C'<\tilde C$ so narrow that $\hat S$ is spacelike for $C'$. Actually, $C'$ can be chosen so narrow that $\hat S$ is  acausal  for $C'$. Indeed, let us consider a sequence of locally Lipschitz proper cone structures  $C<C_n<\tilde C$, such that $\hat S$ is spacelike for every $C_n$ and $C_n\to C$ pointwise. Let $H$ be a stable Cauchy hypersurface for $\tilde C$ such that $\hat S\cap \tilde J^-(H)=\emptyset$. If the statement were false we could find a sequence of continuous $C_n$-causal curves $\sigma_n$ starting and ending at $\hat S$. The curves cannot contract to a point (due to \cite[Thm. 2.33]{minguzzi17} applied to $\tilde S$), thus by the limit curve theorem and the acausality of $\hat S$ in $(M,C)$, there is a past inextendibile limit continuous $C$-causal curve $\sigma^q$, ending at some point $q\in \hat S$. But then $\sigma^q$ intersects $H$ and hence enters $ J^-(H)\subset \tilde J^-(H)$ and so do all $\sigma_n$ for sufficiently large $n$, a contradiction with $\tilde J^+(\hat S)\cap H=\emptyset$.

So let $C'>C$ be as established, and let us choose it smooth. Let $h$ be a complete Riemannian metric. Moreover, let $\mathscr{F}'$, $\mathscr{F}'{}^{-1}(0)=\p C'$, be a Lorentz-Finsler fundamental function, as regular  as it is described in Sec.\ \ref{nog} and with $\mathscr{F}'$ so large that  $\mathscr{F}'{}^{-1}(1)\cap C$ is contained in the unit ball of $4h$. As a consequence, for every $v\in C$, $2 \Vert v \Vert_h \le  \mathscr{F}'(v)$.
 Let $H$ be a stable Cauchy hypersurface for $C'$ such that $\hat S\cap J'{}^-(H)=\emptyset$. Since $\hat S$ is compact, the set $J'{}^+(H)\cap J'{}^{-1}(\hat S)$ is compact and hence cannot future imprison any continuous $C'$-causal curve (remember that $(M,C')$ is globally hyperbolic hence non-total imprisoning). It follows that $L=\p[J'{}^-(H)\cup J'{}^-(\hat S)]$ is a Cauchy hypersurface for $(M,C')$
such that $\hat S \subset L$, due to the $C'$-acausality of $\hat S$.

The function $t\colon M\to \mathbb{R}$, $t=d'_L$, satisfies Eq.\ (\ref{jpo}) over every $C'$-timelike curve $x\colon [0,1]\to M$.
Let $U\subset D(\hat S\backslash \textrm{edge} \hat S)$ be a neighborhood of $S$ so small that  $t\vert_U =d'_{\hat S\backslash \textrm{edge} \hat S}\vert_U$ has the same regularity as $\hat S\backslash \textrm{edge} \hat S$ and hence as $S$, cf.\ Thm.\ \ref{nja}.
In particular, $t\vert_U$ is $C^1$ thus using the inequality (\ref{jpo})  we have on $U$, (just consider a $\mathscr{F}'$-geodesic with tangent $v$)  $\dd t(v)\ge \mathscr{F}'(v)$ for every $C'$-timelike vector $v$, and hence the inequality holds  for every $v\in C$.
 We conclude that $\dd t(v)\ge \mathscr{F}'(v)-\Vert v \Vert_h$  for every  $v\in C\vert_U$.

Let $A,V$ with $A\subset  V\subset U$, be closed and open neighborhoods of $S$, respectively. Let $\alpha\colon M\to \mathbb{R}$ be a positive continuous function. By Theorem \ref{moz} there is a continuous function $\tau$, $\vert \tau-t\vert<\alpha$, that coincides with $t$ on $A$, has the regularity of $t$ and hence of $S$ on $V$, is smooth outside $V$,
 and satisfies  $\dd \tau(v)\ge \mathscr{F}'(v)-\Vert v \Vert_h$ for every $v\in C$. But  $\mathscr{F}'(v)- \Vert v\Vert_h  \ge \Vert v \Vert_h$, thus the function $\tau$ is Cauchy due to the $h$-steepness and the fact that $h$ is complete. Notice that since $\tau$  coincides with $t$ on $A$, $\Sigma:=\tau^{-1}(0)\supset S$, and $\Sigma$ being the level set of a $C^k$ (smooth outside $V$) Cauchy temporal function is a $C^k$ Cauchy hypersurface (smooth outside $V$).
\end{proof}

\subsection{Smooth separation and density theorems}

The sets introduced with the following definition are called {\em trapping domains} in the terminology of \cite{bernard18}.
\begin{definition}
Let $(M,C)$ be a closed cone structure. We say that $F$ is a {\em stable future set} if there is a locally Lipschitz proper cone structure $\tilde C>C$ such that $\tilde I^+(F)\subset F$. Stable past sets are defined analogously.
\end{definition}

Notice that if a stable future set $F$ is open, then $\tilde I^+(F)=F$, with $\tilde C$ as above \cite[Prop.\ 2.13]{minguzzi17}. If $F$ is a stable future set then $\tilde I^+(F)$ is an open stable future set because $\tilde I^+(\tilde I^+(F))=\tilde I^+(F)$. Since for a locally Lipschitz proper cone structure $\tilde C$ the identity $\tilde I\circ \tilde J\cup \tilde J\circ \tilde I\subset \tilde I$ holds true \cite[Thm.\ 2.7]{minguzzi17}, all the standard results for future sets in Lorentzian geometry hold true \cite[Prop.\ 2.84]{minguzzi18b}.

\begin{theorem} \label{jpa}
Let $(M,C)$ be a stably causal closed cone structure  and let $A$ and $B$ be a stable past set and a stable future set, respectively, such that $A\cap B=\emptyset$. Then there is a smooth isotone function $t\colon M\to [0,1]$, temporal on $M\backslash (A\cup B)$, such that $t^{-1}(0)=A$ and $t^{-1}(1)=B$.
\end{theorem}

Observe that the sets $A$ and $B$ can be empty.

\begin{proof}
Let $\tilde C>C$ be a causal locally Lipschitz proper cone structure  such that $\tilde I^-(A)\subset A$ and $\tilde I^+(B)\subset B$, and let $\hat C$, $C<\hat C<\tilde C$, be an intermediate locally Lipschitz proper cone structure. Then \cite[Prop.\ 2.15]{minguzzi17} $\hat J_S\subset \Delta \cup \tilde I$, thus $\hat J^-_S(A)\subset A$ and $\hat J_S^+(B)\subset B$. Since the Seifert relation $\hat J_S$ is closed, we know from \cite[Cor.\ 2.8, Thm.\ 5.3]{minguzzi11f} that $(M,\hat J_S)$ is perfectly normally ordered, that is, there is a continuous $\hat J_S$-isotone (and hence $\hat J$-isotone) function $\tau\colon M\to [0,1]$ such that $\tau^{-1}(0)=A$ and $\tau^{-1}(1)=B$.

Let $\varphi\colon M\to [0,1]$ be a smooth temporal function for $(M,\hat C)$. Let $\epsilon>0$. The function $\epsilon \varphi$ is also smooth temporal, thus there is a  Riemannian metric $g$ such that  $\epsilon \varphi$ is $g$-steep on $(M,\hat C)$, just take the unit balls of $g$ so large on $TM$ that they contain $(\dd \epsilon\varphi)^{-1} (1) \cap \hat C$.
 The function $\tilde \tau=\tau+\epsilon \varphi$ is  $g$-anti Lipschitz on $(M,\hat C)$, since the term $\epsilon  \varphi$ has this property  (this trick, due to Fathi, was used in \cite{chrusciel13} and \cite{bernard18b}). Thus $\tilde \tau$ can be smoothed using Thm.\ \ref{moz} with $\underline{F}(v)=\Vert v\Vert_{g}$, $h=g/4$, $\alpha=\epsilon$. We can find a smooth temporal (for $C$) function $\hat \tau$, $\vert \hat \tau-\tilde \tau\vert<\epsilon$, and hence $\vert \hat \tau-\tau\vert<2\epsilon$.
Let us consider a smooth function $\phi_\epsilon\colon \mathbb{R}\to [0,1]$ which is 0 on $(-\infty,2\epsilon]$, and 1 on $[1-2\epsilon,+\infty)$ and such that $\phi_\epsilon'>0$ elsewhere, then $\phi_\epsilon \circ \hat \tau(A)=0$, $\phi_\epsilon \circ \hat \tau(B)=1$, $(\phi_\epsilon \circ \hat \tau)^{-1}(0)\subset \tau^{-1}([0,4\epsilon])$, $(\phi_\epsilon \circ \hat \tau)^{-1}(1)\subset \tau^{-1}([1-4\epsilon,1])$. Notice that $\phi_\epsilon \circ \hat \tau$ is temporal wherever it differs from 0 and 1.
Let $f_k=\phi_{1/k} \circ \hat \tau$, and observe that for every $p\in M\backslash (A\cup B)$ we can find a sufficiently large $k$ such that $4/k<\tau(p)<1-4/k$, so that $0<f_k(p)<1$ and hence $f_k$ is temporal at $p$. There are constants $c_k>0$ such that, $\sum_k c_k=1$ and $t=\sum_kc_kf_k$ is smooth \cite[Lemma 3.2]{fathi97}, which concludes the proof.
\end{proof}




The previous separation theorem can be improved as follows.

\begin{theorem} \label{mos}
Let $(M,C)$ be a stably causal closed cone structure. Let $A$ be a closed $J_S$-past set, $J^-_S(A)\subset A$, and let  $B$ be a closed $J_S$-future set, $J^+_S(B)\subset B$, which are disjoint. Then there is a smooth isotone function $t\colon M\to [0,1]$, temporal on $M\backslash (A\cup B)$ such that $A=t^{-1}(0)$ and $B=t^{-1}(1)$.
\end{theorem}

\begin{proof}
Since the Seifert relation $J_S$ is closed, we know from \cite[Cor.\ 2.8, Thm.\ 5.3]{minguzzi11f} that $(M,J_S)$ is perfectly normally ordered, that is, there is a continuous $J_S$-isotone function $\tau\colon M\to [0,1]$ such that $\tau^{-1}(0)=A$ and $\tau^{-1}(1)=B$.


Let $\epsilon>0$ and let $\varphi\colon M\to (0,1)$ be a time function for $(M,C)$, then $\hat \tau:=\tau+\epsilon
\varphi$ is also a time function.
Let $h$ be a complete Riemannian metric such that
\begin{align*}
d_h\big( \hat \tau^{-1}((-\infty,\epsilon]), \hat \tau^{-1}([2\epsilon,+\infty))\big)&>\epsilon, \\ d_h\big(\hat\tau^{-1}((-\infty,1-2\epsilon]), \hat\tau^{-1}([1-\epsilon, +\infty))\big)&>\epsilon.
\end{align*}
 By \cite[Lemma 3.5]{minguzzi17} we can find a locally Lipschitz proper cone structure $\tilde C>C$ such that the open $\tilde C$-future  set $V=\tilde I^+(\hat \tau^{-1}(2\epsilon,+\infty))$ is disjoint from $\hat \tau^{-1}((-\infty,\epsilon])\supset A$.

 Notice that $\tilde A:=M\backslash V$ is a closed $\tilde C$-past set such that
\[
A\subset \tilde A\subset \hat \tau^{-1}((-\infty,2\epsilon])\subset \tau^{-1}([0,3\epsilon]).
\]
 In other words, $\tilde A$ is a stable past set for $(M,C)$. Dually, $\tilde C>C$ can be further narrowed so that the open $\tilde C$-past  set $U=\tilde I^-(\hat \tau^{-1}(-\infty, 1-2\epsilon))$ is disjoint from $\hat \tau^{-1}([1-\epsilon, +\infty))\supset B$.  Notice that $\tilde B=M\backslash U$ is a closed $\tilde C$-future set such that
\[
B\subset \tilde B\subset \hat\tau^{-1}( [1-2\epsilon, +\infty))\subset \tau^{-1}([1-3\epsilon,1]).
\]
In other words, $\tilde B$ is a stable future set for $(M,C)$. By Thm.\ \ref{jpa}
there is a smooth isotone function $t_\epsilon\colon M\to [0,1]$, temporal on $M\backslash (\tilde A\cup \tilde B)$, such that $t_\epsilon^{-1}(0)=\tilde A$ and $t_\epsilon^{-1}(1)=\tilde B$. There are constants $c_k>0$ such that $\sum_k c_k=1$ and $t=\sum_k c_k t_{1/k}$ is smooth, which concludes the proof.
\end{proof}

We recall  \cite[Thm.\ 3.14]{minguzzi17} that if $t$ is a time function for a closed cone structure we have: $(p,q)\in J_S\backslash \Delta \Rightarrow t(p)<t(q)$. The following density result improves the bound in \cite[Cor.\ 9]{bernard18b}. Interestingly, our proof is completely different.

\begin{theorem}  \label{jif}
Let $t$ be a continuous isotone function on the stably causal closed cone structure $(M,C)$ and let $\epsilon\colon \mathbb{R}\to (0,+\infty)$ be a smooth positive function. Then there is a smooth temporal function $\tau$ such that for every $p\in M$, $\vert \tau(p)-t(p)\vert<\epsilon(t(p))$.
\end{theorem}

Clearly, if $t$ is Cauchy and $\epsilon$ is chosen bounded then $\tau$ is Cauchy.

\begin{proof}
Since $(M,C)$ is stably causal there is a time function $g\colon M\to (0,1)$. Let $\sigma\colon \mathbb{R}\to \mathbb{R}$ be a smooth function such that $\sigma'$ is increasing and
\[
2 \max(\sup{}_{[t,t+1]}\epsilon^{-1},1) <\sigma'(t).
\]
 The function $\tilde t=\sigma\circ t+g$ is a time function.

By \cite[Thm.\ 3.14]{minguzzi17} for each interval $[a,b]$ the set $\tilde t^{-1}((-\infty, a])$ is a closed $J_S$-past set and the set $\tilde t^{-1}([b, +\infty))$ is a $J_S$-future set.  By Thm.\ \ref{mos} we can find a smooth function $\varphi$ which is zero on the former set, one on the latter set, while with value in $(0,1)$ and temporal elsewhere. Then $\tau\colon \tilde t^{-1}([a,b])\to [a,b]$, $\tau=a+(b-a)\varphi$, is smooth, has value $a$ on $\tilde t^{-1}(a)$, $b$ in $\tilde t^{-1}(b)$, while it has value in $(a,b)$ and is temporal elsewhere. Thus we can redefine $\tilde t$ over $\tilde t^{-1}(\mathbb{R})=\cup_{k \in \mathbb{Z}} \tilde t^{-1}
\big(\frac{1}{2} [2k,2k+2]\big)$
by redefining it in the preimage of each single interval, so as to obtain a smooth function $\tilde  t_{even}$, such that $\vert \tilde t_{even}-\tilde t\vert<1$  which is temporal everywhere save for the points in $\cup_k \tilde t^{-1}(k)$. Similarly, writing $\tilde t^{-1}(\mathbb{R})=\cup_{k \in \mathbb{Z}} \tilde t^{-1}(\frac{1}{2} [2k+1,2k+3])$ we find a smooth function $\tilde t_{odd}$,  such that $\vert \tilde t_{odd}-\tilde t\vert<1$ and which is temporal everywhere save for the points in $\cup_k \tilde t^{-1}(k+\frac{1}{2})$. Then the function $\tau$ such that $\sigma \circ \tau=(\tilde t_{even}+\tilde t_{odd})/2$,  is the desired smooth temporal function. In fact, $\sigma\circ \tau$ is temporal (so $\tau$ is) and $\vert\sigma \circ \tau-\sigma\circ t  \vert\le\vert\sigma \circ \tau- \tilde t\vert+\vert \tilde t-\sigma\circ t  \vert <2$. If $t(p)\le \tau(p)$, $\sigma'(t(p))[\tau(p)-t(p)]\le \vert\sigma(\tau(p))-\sigma(t(p))  \vert<2$, which implies $\tau(p)-t(p)<\epsilon(t(p))$. If $t(p)>\tau(p)$ then since $\sigma'>2$, and $\sigma(t(p))-\sigma(\tau(p))<2$, we have $t(p)-\tau(p)<1$, thus $\sigma'(\tau(p))>2 \epsilon^{-1}(t(p))$, and we get $\sigma'(\tau(p))[t(p)-\tau(p)]\le \vert\sigma(\tau(p))-\sigma(t(p))  \vert<2$, which implies $t(p)-\tau(p)<\epsilon(t(p))$.
\end{proof}

Under stable causality we can find temporal functions adapted to stably acausal spacelike hypersurfaces.

\begin{theorem} \label{jcy}
Let $(M,C)$ be a stably causal closed cone structure. Let $S$ be a $C^{1,1}$ or a $C^k$, $2\le k\le \infty$ stably acausal spacelike hypersurface without edge (hence a partial Cauchy hypersurface), then there is a  smooth temporal function $t$ of the same regularity such that $S\subset t^{-1}(0)$.
\end{theorem}

The stable acausality condition cannot be deduced from the other assumptions and is necessary (because the cones can be widened preserving the temporality of $t$).

\begin{proof}
 Let $\tilde C>C$ be a stably causal locally Lipschitz proper cone structure such that $S$ is acausal for $\tilde C$.
Let $\hat C$, $C<\hat C<\tilde C$ be an intermediate locally Lipschitz proper cone structure. Since $S$ is $\tilde C$-achronal there is a maximal $\tilde C$-achronal set $\Sigma$, $M=\Sigma \cup \tilde I^+(\Sigma)\cup \tilde I^-(\Sigma)$. The removal of $\Sigma$ disconnects the spacetime.
 The $\tilde C$-achronality of $\Sigma$ implies that it is $\hat C$-acausal \cite[Thm.\ 2.24]{minguzzi17}. The maximal achronality  implies that it does not have edge and that it is a topological hypersuface (see also  \cite[Thm.\ 2.87]{minguzzi18b}).
On the globally hyperbolic spacetime $\hat D(\Sigma)$ we consider a signed distance $d'_{\Sigma}$ for some locally Lipschitz proper Lorentz-Finsler structure with $\hat C<C'<\tilde C$. Then we proceed as in the proof of Thm.\ \ref{ngy} but with $A$, $V$, $U$, neighborhoods of $S$  (rather than $\Sigma$) all contained in a neighborhood of $S$ where $d'_S$ is $C^k$. Thus we obtain a $C^k$ Cauchy temporal function $\tau$ such that $\tau^{-1}(0)\supset S$. Let $\sigma(x)$,  be a smooth non-decreasing function equal to $x$ in a neighborhood of 0, equal to 1 for $x\ge 1$, and equal to $-1$ for $x\le -1$. The function $\sigma\circ \tau$ can be extended all over $M$ to a $C^k$ isotone function by setting $\sigma\circ \tau=\pm 1$ in $\tilde I^\pm(\Sigma)\backslash \hat D(\Sigma)$. The set $P:=\{q: \tau(q)\le 0\}$ is a stably past set thus by Thm.\ \ref{jpa} there is a smooth function $\tau^+$, temporal on $M\backslash P$, which vanishes on $P$. Similarly there is a smooth function $\tau^-$ that vanishes on the stable future set $F:=\{q: \tau(q)\ge 0\}$ and is temporal in $M\backslash F$. The function $t=\tau+\tau^++\tau^-$ is the desired smooth temporal function.
\end{proof}

By using the same set of ideas it is possible to prove the following result.
\begin{theorem} \label{jky}
Let $(M,C)$ be a stably causal closed cone structure. Let $S_0$ and $S_1$ be two $C^{1,1}$ or two $C^k$, $2\le k\le \infty$, stably acausal spacelike hypersurfaces without edge (hence  partial Cauchy hypersurfaces), such that $S_1\subset J^+_S(S_0)\backslash S_0$. Then there is a  smooth temporal function $t$ of the same regularity such that $S_0\subset t^{-1}(0)$ and $S_1\subset t^{-1}(1)$.
\end{theorem}

In the proof one has to sum five smooth isotone functions which are temporal on suitable open sets.

\section{Conclusions}
By using an approach complementary to that proposed by Bernard and Suhr \cite{bernard19}, we
 generalized the  results by Bernal and S\'anchez  \cite{bernal06}  to the framework of closed cone structures.
Preliminarily we showed that the spacelike Cauchy hypersurfaces are stable and that the  signed distance $d_S$ from a spacelike hypersurface $S$ is, in a neighborhood of it, as regular as the  hypersurface. Then we applied our smoothing theorem to the signed distance function of a Lorentz-Finsler structure having wider cones. The statements derived from our approach are rather detailed, in fact they are useful already in the Lorentzian metric theory, and the  proofs are conceptually rather straightforward.  Furthermore, they nicely  pass through concepts, such as the signed distance, that deserved to be investigated.

\section*{Acknowledgements}
I thank James Vickers and G\"unther H\"ormann for some interesting questions that motivated this investigation. Work partially supported by GNFM of INDAM.


\begin{thebibliography}{10}

\bibitem{bernal03}
A.~N. Bernal and M.~{S\'a}nchez.
\newblock On smooth {C}auchy hypersurfaces and {G}eroch's splitting theorem.
\newblock {\em Commun. Math. Phys.}, 243:461--470, 2003.

\bibitem{bernal06}
A.~N. Bernal and M.~{S\'a}nchez.
\newblock Further results on the smoothability of {C}auchy hypersurfaces and
  {C}auchy time functions.
\newblock {\em Lett. {M}ath. {P}hys.}, 77:183--197, 2006.

\bibitem{bernard18}
P.~Bernard and S.~Suhr.
\newblock Lyapounov functions of closed cone fields: from {C}onley theory to
  time functions.
\newblock {\em Commun. Math. Phys.}, 359:467--498, 2018.

\bibitem{bernard18b}
P.~Bernard and S.~Suhr.
\newblock Smoothing causal functions.
\newblock {\em J. of Phys.: Conf. Series}, 968:012001, 2018.

\bibitem{bernard19}
P.~Bernard and S.~Suhr.
\newblock Cauchy and uniform temporal functions of globally hyperbolic cone
  fields.
\newblock {\em Proc. Amer. Math. Soc.}, 2020.
\newblock In press. {arXiv:}1905.06006.

\bibitem{chrusciel13}
P.~T. Chru{\'s}ciel, J.~D.~E. Grant, and E.~Minguzzi.
\newblock On differentiability of volume time functions.
\newblock {\em Ann. Henri {P}oincar{\'e}}, 17:2801--2824, 2016.
\newblock {arXiv}:1301.2909.

\bibitem{fathi97}
A.~Fathi.
\newblock Partitions of unity for countable covers.
\newblock {\em Amer. Math. Monthly}, 104:720--723, 1997.

\bibitem{fathi12}
A.~Fathi and A.~Siconolfi.
\newblock On smooth time functions.
\newblock {\em Math. {P}roc. {C}amb. {P}hil. {S}oc.}, 152:303--339, 2012.

\bibitem{hawking73}
S.~W. Hawking and G.~F.~R. Ellis.
\newblock {\em The Large Scale Structure of Space-Time}.
\newblock Cambridge {U}niversity {P}ress, Cambridge, 1973.

\bibitem{hormann19}
G.~{H\"o}rmann, Y.~Sanchez Sanchez, C.~Spreitzer, and J.~A. Vickers.
\newblock Green operators in low regularity spacetimes and quantum field
  theory.
\newblock {\em Class. Quantum Grav.}, 2020.
\newblock In press https://doi.org/10.1088/1361-6382/ab839a.

\bibitem{minguzzi19c}
R.~A. Hounnonkpe and E.~Minguzzi.
\newblock Globally hyperbolic spacetimes can be defined without the `causal'
  condition.
\newblock {\em Class. Quantum Grav.}, 36:197001, 2019.
\newblock {arXiv}:1908.11701.

\bibitem{javaloyes18}
M.~A. Javaloyes and M.~S\'anchez.
\newblock On the definition and examples of cones and {F}insler spacetimes.
\newblock {\em RACSAM}, 114:30, 2020.

\bibitem{krantz81}
S.~G. Krantz and H.~R. Parks.
\newblock Distance to {$C^k$} hypersurfaces.
\newblock {\em J. Diff. Geom.}, 40:116--120, 1981.

\bibitem{kunzinger13}
M.~Kunzinger, R.~Steinbauer, and M.~Stojkovi\'c.
\newblock The exponential map of a {$C^{1,1}$}-metric.
\newblock {\em Differential Geom. Appl.}, 34:14--24, 2014.
\newblock arXiv:1306.4776v1.

\bibitem{matsumoto92}
K.~Matsumoto.
\newblock A note on the differentiability of the distance function to regular
  submanifolds of {R}iemannian manifolds.
\newblock {\em Nihonkai Math. J.}, 3:81--85, 1992.

\bibitem{minguzzi11f}
E.~Minguzzi.
\newblock Normally preordered spaces and utilities.
\newblock {\em Order}, 30:137--150, 2013.
\newblock {arXiv}:1106.4457v2.

\bibitem{minguzzi13d}
E.~Minguzzi.
\newblock Convex neighborhoods for {L}ipschitz connections and sprays.
\newblock {\em Monatsh. Math.}, 177:569--625, 2015.
\newblock {arXiv}:1308.6675.

\bibitem{minguzzi14h}
E.~Minguzzi.
\newblock An equivalence of {F}inslerian relativistic theories.
\newblock {\em Rep. Math. Phys.}, 77:45--55, 2016.
\newblock {arXiv}:1412.4228.

\bibitem{minguzzi15e}
E.~Minguzzi.
\newblock Affine sphere relativity.
\newblock {\em {C}ommun. {M}ath. {P}hys.}, 350:749--801, 2017.
\newblock {arXiv}:1702.06739.

\bibitem{minguzzi17}
E.~Minguzzi.
\newblock Causality theory for closed cone structures with applications.
\newblock {\em Rev. Math. Phys.}, 31:1930001, 2019.
\newblock {arXiv}:1709.06494.

\bibitem{minguzzi18b}
E.~Minguzzi.
\newblock Lorentzian causality theory.
\newblock {\em Living Rev. Relativ.}, 22:3, 2019.
\newblock {https}://doi.org/10.1007/s41114-019-0019-x.

\end{thebibliography}

\end{document}